\documentclass[a4paper,USenglish,numbers=noenddot]{scrartcl}
\def\doi#1{\href{https://doi.org/\detokenize{#1}}{\url{https://doi.org/\detokenize{#1}}}}

\usepackage{amssymb,amsmath,mathrsfs,amsthm}
\usepackage[algo2e, vlined, ruled, linesnumbered]{algorithm2e}
\usepackage{tikz}

\usepackage[hidelinks]{hyperref}
\usepackage[capitalize,nameinlink]{cleveref}
\usepackage[shortlabels]{enumitem}
\setlist[enumerate]{leftmargin=1cm}

\newtheorem{theorem}{Theorem}
\newtheorem{lemma}[theorem]{Lemma}
\newtheorem{definition}[theorem]{Definition}
\newtheorem{corollary}[theorem]{Corollary}

\newtheorem{observation}[theorem]{Observation}

\newcommand{\+}{\textsuperscript{+}}
\newcommand{\G}{\ensuremath{\mathcal{G}}}
\newcommand{\N}{\ensuremath{\mathbb{N}}}

\newcommand{\B}{\ensuremath{\mathcal{B}}}
\newcommand{\I}{\ensuremath{\mathcal{I}}}
\newcommand{\NP}{\ensuremath{\mathsf{NP}}}
\renewcommand{\P}{\ensuremath{\mathsf{P}}}
\newcommand{\A}{\ensuremath{\mathcal{A}}}
\renewcommand{\O}{\ensuremath{{\cal O}}}
\newcommand{\cf}{\ensuremath{{\cal F}}}
\newcommand{\cl}{\ensuremath{{\cal L}}}

\newcommand\mdoubleplus{\ensuremath{\mathbin{+\mkern-3mu+}}}

\tikzstyle{treeedge}=[line width=1mm]

\begin{document}

\title{Graph Search Trees and Their Leaves}
	
	\author{Robert Scheffler}

	\date{Institute of Mathematics, Brandenburg University of Technology, Cottbus, Germany \\ \texttt{robert.scheffler@b-tu.de}}

  	\maketitle

 	\begin{abstract}
		Graph searches and their respective search trees are widely used in algorithmic graph theory. The problem whether a given spanning tree can be a graph search tree has been considered for different searches, graph classes and search tree paradigms. Similarly, the question whether a particular vertex can be visited last by some search has been studied extensively in recent years. We combine these two problems by considering the question whether a vertex can be a leaf of a graph search tree. We show that for particular search trees, including DFS trees, this problem is easy if we allow the leaf to be the first vertex of the search ordering. We contrast this result by showing that the problem becomes hard for many searches, including DFS and BFS, if we forbid the leaf to be the first vertex. Additionally, we present several structural and algorithmic results for search tree leaves of chordal graphs.
	\end{abstract}

\section{Introduction}

Graph searches are an extensively used concept in algorithmic graph theory. The searches BFS and DFS belong to the most basic algorithms and are used in a wide range of applications as subroutines. The same holds for more sophisticated searches as LBFS, LDFS, and MCS (see, e.g., \cite{berry2004maximum,corneil2013ldfs,corneil2009lbfs}).

An important structure closely related to a graph search is the corresponding search tree. Such a tree contains all the vertices of the graph and for every vertex different from the start vertex exactly one edge to a vertex preceding it in the search ordering. Those trees can be of particular interest as for instance the tree obtained by a BFS contains the shortest paths from the root to all other vertices in the graph and DFS trees are used for fast planarity testing~\cite{hopcroft1974efficient}. Furthermore, trees generated by LBFS were used to design a linear-time implementation of the search LDFS for chordal graphs~\cite{beisegel2020linear}.

The problem of deciding whether a given spanning tree of a graph can be obtained by a particular search was introduced by Hagerup~\cite{hagerup1985biconnected} in 1985, who presented a linear-time algorithm that recognizes DFS trees. In the same year, Hagerup and Nowak~\cite{hagerup1985recognition} gave a similar result for the BFS tree recognition. In 2021, Beisegel~et~al.~\cite{beisegel2021recognition} presented a more general framework for the search tree recognition problem. They introduced the term \cf-tree for search trees where a vertex is connected to its first visited neighbor, i.e., BFS-like trees, and \cl-trees for search trees where a vertex is connected to its most recently visited neighbor, i.e., DFS-like trees. They showed, among other things, that \cf-tree recognition is \NP-hard for LBFS, LDFS, and MCS on weakly chordal graphs, while the problem can be solved in polynomial time for all three searches on chordal graphs. These results are complemented in~\cite{scheffler2022recognition}, where it is shown that the recognition of \cf-trees of DFS and \cl-trees of BFS is \NP-hard, a strong contrast to the polynomial results for \cf-trees of BFS and \cl-trees of DFS.

Another feature of a graph search that was used several times within algorithms are its end-vertices, i.e., the vertices that can be visited last by the search. Some of these end-vertices have nice properties. One example are the end-vertices of LBFS on chordal graphs. These vertices are simplicial, a fact that was used by Rose et al.~\cite{rose1976algorithmic} to design a linear-time recognition algorithm for chordal graphs. Furthermore, the
end-vertices of LBFS are strongly related to dominating pairs of AT-free graphs~\cite{corneil1999linear} and transitive orientations of comparability graphs~\cite{habib2000lex-bfs}. Thus, it is well motivated to consider the end-vertex problem, i.e., the question whether a given vertex of a graph is an end-vertex of a particular search. Introduced in 2010 by Corneil et~al.~\cite{corneil2010end}, the problem has gained much attention by several researchers, leading to a wide range of hardness results and algorithms for different searches on different graph classes (see, e.g.,~\cite{beisegel2019end-vertex,charbit2014influence,gorzny2017end-vertices,kratsch2015end,rong2022graph,zou2022end}).

If we compare the known complexity results for the end-vertex problem and the recognition problem of \cf-trees, we notice strong similarities between these two problems. Motivated by that fact, a generalization of both problems, called \emph{Partial Search Order Problem}, was introduced in~\cite{scheffler2022linearizing}. This problem asks whether a given partial order on a graph's vertex set can be linearly extended by a search ordering. Another way to combine the end-vertex problem with the search tree recognition problems is motivated by the following observation: If a vertex is the end-vertex of some search ordering, then it is a leaf in the respective search tree, no matter whether we consider the \cf-tree or the \cl-tree. Therefore, we ask whether a given vertex can be a leaf of a search tree constructed by a particular search. Note that this problem was first suggested in 2020 by Michel Habib. Here, we study its complexity for $\cf$-trees and $\cl$-trees of several searches, including BFS, DFS, LBFS, LDFS, and MCS.

\paragraph{Our Contribution.} We consider two different types of leaves of search trees. A leaf is a \emph{root leaf} of a search tree if it is the start vertex of the respective search ordering. All other leaves of a search tree are called \emph{branch leaves}. We show that it is easy for all the searches considered here to identify the possible root leaves both for \cf-trees and for \cl-trees. For some searches, including DFS, these results imply directly that the general problem of recognizing leaves of \cl-trees is easy. This is contrasted by the result that, at least for DFS, the recognition of branch leaves of \cl-trees is \NP-hard. We show that the same holds for \cf-tree branch leaves of several searches, including DFS and BFS. In contrast, the leaves of \cl-trees of BFS can be recognized in polynomial time for bipartite graphs. This is quite surprising since the \cl-tree recognition problem of BFS is \NP-hard on bipartite graphs~\cite{scheffler2022recognition} while \cf-trees of BFS can be recognized efficiently on general graphs~\cite{hagerup1985recognition}. In the final section we consider chordal graphs and show that on this graph class the branch leaves of almost all considered searches can be recognized in linear time.

\section{Preliminaries}

\subsection{General Notation}
The graphs considered in this paper are finite, undirected, simple and connected. Given a graph $G$, we denote by $V(G)$ the \emph{set of vertices} and by $E(G)$ the \emph{set of edges}. The terms $n(G)$ and $m(G)$ describe the number of vertices and edges of $G$, respectively, i.e., $n(G) = |V(G)|$ and $m(G) = |E(G)|$. For a vertex $v\in V(G)$, we denote by $N_G(v)$ the \emph{(open) neighborhood} of $v$ in $G$, i.e., the set $N_G(v)=\{u\in V\mid uv\in E\}$ where $uv$ denotes an edge between $u$ and $v$. The \emph{closed neighborhood} of a vertex $v$ is the union of the open neighborhood of $v$ with the set $\{v\}$ and is denoted by $N_G[v]$. Given a set $S \subseteq V(G)$, the term $G[S]$ describes the subgraph of $G$ that is induced by $S$.

The \emph{distance} $dist_G(v,w)$ of two vertices $v$ and $w$ in $G$ is the length (i.e., the number of edges) of the shortest $v$-$w$-path in $G$. The \emph{eccentricity} $ecc_G(v)$ of a vertex $v$ in $G$ is the largest distance of $v$ to any other vertex in $G$. The \emph{diameter} $diam(G)$ of $G$ is the largest eccentricity of a vertex in $G$ and the \emph{radius} $rad(G)$ of $G$ is the smallest eccentricity of a vertex in $G$. A vertex $v$ with $ecc_G(v) = rad(G)$ is called \emph{central vertex} of $G$. The set $N_G^\ell(v)$ contains all vertices whose distance to the vertex $v$ in $G$ is equal to $\ell$.

A \emph{vertex ordering} of $G$ is a bijection $\sigma: \{1,2,\dots,|V(G)|\}\rightarrow V(G)$. We denote by $\sigma^{-1}(v)$ the position of vertex $v\in V(G)$. Given two vertices $u$ and $v$ in $G$ we say that $u$ is \emph{to the left} (resp. \emph{to the right}) of $v$ if $\sigma^{-1}(u)<\sigma^{-1}(v)$ (resp. $\sigma^{-1}(u)>\sigma^{-1}(v)$) and we denote this by $u \prec_{\sigma}v$ (resp.  $u \succ_{\sigma}v$). 
Given two orderings $\sigma$ of $X$ and $\rho$ of $Y$ with $X \cap Y = \emptyset$, the ordering $\tau = \sigma \mdoubleplus \rho$ is the \emph{concatenation} of $\sigma$ and $\rho$, i.e., $\tau(i) = \sigma(i)$ if $1 \leq i \leq |X|$ and $\tau(i) = \rho(i - |X|)$ if $|X| < i \leq |X \cup Y|$. If $v \notin X$, then $v \mdoubleplus \sigma$ (or $\sigma \mdoubleplus v$) denotes the concatenation of $\sigma$ with the linear ordering of the set $\{v\}$. 

A \emph{clique} in a graph $G$ is a set of pairwise adjacent vertices and an \emph{independent set} in $G$ is a set of pairwise nonadjacent vertices. A clique $C$ is \emph{dominating} if any vertex of $G$ is either in $C$ or has a neighbor in $C$. A vertex $v$ is \emph{simplicial} if its neighborhood induces a clique. A vertex $v$ of a connected graph $G$ is a \emph{cut vertex} if $G - v$ is not connected. Two vertices $u$ and $w$ form a \emph{two-pair} if any induced path between $u$ and $w$ has length two.

A graph is \emph{bipartite} if its vertex set can be partitioned into two independent sets $X$ and $Y$. A graph is \emph{weakly chordal} if $G$ contains neither an induced cycle of the length $\geq 5$ nor the complement of such an induced cycle. A graph is \emph{chordal} if it does not contain an induced cycle of length $\geq 4$. A vertex ordering $\sigma$ of a graph $G$ is a \emph{perfect elimination ordering} if any vertex $v$ is simplicial in the graph $G[S(v)]$ with $S(v) := \{w~|~w \prec_\sigma v\}$. A graph $G$ has a PEO if and only if $G$ is chordal~\cite{rose1970triangulated}. A \emph{split graph} $G$ is a graph whose vertex set can be partitioned into sets $C$ and $I$, such that $C$ is a clique in $G$ and $I$ is an independent set in $G$. It is easy to see that any split graph is chordal. 

A \emph{tree} is an acyclic connected graph. A \emph{spanning tree} of a graph $G$ is an acyclic connected subgraph of $G$ which contains all vertices of $G$. A tree together with a distinguished \emph{root vertex} $r$ is said to be \emph{rooted}. In such a rooted tree $T$, a vertex $v$ is an \emph{ancestor} of vertex $w$ if $v$ is an element of the unique path from $w$ to the root $r$. A vertex $ w $ is called the \emph{descendant} of $ v $ if $ v $ is the ancestor of $ w $. Vertex $v$ is the \emph{parent} of vertex $w$ if $v$ is an ancestor of $w$ and is adjacent to $w$ in $T$. Vertex $w$ is called the \emph{child} of $ v $ if $ v $ is the parent of $ w $.

\subsection{Searches, Search Trees and Their Leaves}
In the most general sense, a \emph{graph search} $\A$ is a function that maps every graph $G$ to a set $\A(G)$ of vertex orderings of $G$. The elements of the set $\A(G)$ are the \emph{$\A$-orderings of $G$}. The graph searches considered in this paper can be formalized adapting a framework introduced by Corneil~et~al.~\cite{corneil2016tie} (a similar framework is given in~\cite{krueger2011general}). This framework uses subsets of $\N^+$ as vertex labels. Whenever a vertex is numbered, its index in the search ordering is added to the labels of its unnumbered neighbors. The search $\A$ is defined via a strict partial order $\prec_{\cal A}$ on the elements of ${\cal P}(\N^+)$ (see \cref{algo:ls}). The respective $\A$-orderings are exactly those vertex orderings that can be found by this framework using the partial label order $\prec_\A$.

\newcommand{\searchlabel}{\emph{label}}
\begin{algorithm2e}[t]
\small
    \KwIn{A graph $G$} 
    \KwOut{A search ordering $\sigma$ of $G$}
    \Begin{
		\lForEach{$v \in V(G)$}{\searchlabel($v$) $\leftarrow$ $\emptyset$}
		\For{$i$ $\leftarrow$ $1$ \KwTo $n(G)$}{
			\emph{Eligible} $\leftarrow$ $\{x \in V(G)~|~x$ unnumbered and $\nexists$ unnumbered $y \in V(G)$ \\ \mbox{}\phantom{\emph{Eligible} $\leftarrow$ $\{x \in V(G)~|~$}such that \searchlabel($x$) $\prec_{\cal A}$ \searchlabel($y$)$\}$\;
			let $v$ be an arbitrary vertex in \emph{Eligible}\;\label{line:ls}
			$\sigma(i)$ $\leftarrow$ $v$\tcc*{assigns to $v$ the number $i$}
			\lForEach{unnumbered vertex $w \in N(v)$}{\searchlabel($w$) $\leftarrow$ \searchlabel($w$) $\cup$ $\{i\}$}
		}
	}
    \caption{Label Search($\prec_{\cal A}$)}\label{algo:ls}
\end{algorithm2e}

In the following, we define the searches considered in this paper by presenting suitable partial orders $\prec_{\cal A}$ (see~\cite{corneil2016tie}). The \emph{Generic Search} (GS) is equal to the Label Search($\prec_{GS}$) where $A \prec_{GS} B$ if and only if $A = \emptyset$ and $B \neq \emptyset$. Thus, any vertex with a numbered neighbor can be numbered next.

The partial label order $\prec_{BFS}$ for \emph{Breadth First Search} (BFS) is defined as follows: $A \prec_{BFS} B$ if and only if $A = \emptyset$ and $B \neq \emptyset$ or $\min(A) > \min (B)$. For the \emph{Lexicographic Breadth First Search} (LBFS)~\cite{rose1976algorithmic} we consider the partial order $\prec_{LBFS}$ with $A \prec_{LBFS} B$ if and only if $A \subsetneq B$ or $\min(A \setminus B) > \min(B \setminus A)$. Both BFS and LBFS are \emph{layered}, i.e., the sets $N_G^\ell(r)$ are consecutive within orderings starting in $r$. We sometimes use the term \emph{layer} if we refer to a set $N^\ell_G(r)$.

The partial label order $\prec_{DFS}$\index{DFS|textit} for \emph{Depth First Search} (DFS) is defined as follows: $A \prec_{DFS} B$ if and only if $A = \emptyset$ and $B \neq \emptyset$ or $\max(A) < \max (B)$. For the \emph{Lexicographic Depth First Search}~\cite{corneil2008unified} we use the strict partial order $\prec_{LDFS}$ where $A \prec_{LDFS} B$ if and only if $A \subsetneq B$ or $\max(A \setminus B) < \max(B \setminus A)$.

The \emph{Maximum Cardinality Search} (MCS)~\cite{tarjan1984simple} uses the partial order $\prec_{MCS}$ with $A \prec_{MCS} B$ if and only if $|A| < |B|$. The \emph{Maximal Neighborhood Search} (MNS)~\cite{corneil2008unified} is defined using $\prec_{MNS}$ with $A \prec_{MNS} B$ if and only if $A \subsetneq B$. It follows directly from these partial label orders, that any LBFS, LDFS, and MCS ordering is also an MNS ordering. Furthermore, the orderings of all presented searches are GS orderings.

In general, there can occur ties during an application of all the searches considered in this paper. This problem can be solved by considering so-called \emph{+-searches}. Let $\A$ be a graph search. Given an arbitrary vertex ordering $\rho$ of a graph $G$, an $\A$-ordering $\sigma = (v_1, \ldots, v_n)$ of $G$ is the $\A\+(\rho)$-ordering of $G$ if $\sigma$ fulfills the following condition: for any $i \in \{0, \ldots, n-1\}$, the vertex $v_{i+1}$ is the \emph{leftmost} vertex in $\rho$ such that $(v_1, \ldots, v_i, v_{i+1})$ is a prefix of an $\A$-ordering of $G$.\footnote{Note that other authors define the vertex $v_{i+1}$ to be the \emph{rightmost} vertex in $\rho$.}

Searches as BFS and DFS are often used to compute corresponding graph search trees. Beisegel et~al.~\cite{beisegel2021recognition} formalized the different concepts of search trees as follows.

\begin{definition}[Beisegel et~al.~\cite{beisegel2021recognition}]
	Let $\sigma$ be a GS ordering of a connected graph $ G $. The \emph{\cf-tree} of $\sigma$ is the spanning tree of $G$ containing the edge from each vertex $v$ with $\sigma^{-1}(v) > 1$ to its leftmost neighbor in $\sigma$.
	
	The \emph{\cl-tree} of $\sigma$ is the spanning tree containing the edge from each vertex $v$ with $\sigma^{-1}(v) > 1$ to its rightmost neighbor $w$ in $\sigma$ with $ w \prec_\sigma v$.
\end{definition}

We will need the following two lemmas about DFS \cl-trees.

\begin{lemma}[Tarjan~\cite{tarjan1972depth}]\label{lemma:dfstrees}
	Let $G$ be a graph and let $T$ be a spanning tree of $G$. Then $T$ is a DFS \cl-tree of $G$ if and only if for each edge $uv \in E(G)$ vertex $u$ is either an ancestor or a descendant of $v$ in $T$.
\end{lemma}

\begin{lemma}[Beisegel et al.~\cite{beisegel2020linear}]\label{lemma:tree}
Let $T$ be an $\cl$-tree of some DFS ordering of $G$ rooted in $s$ and let $\sigma$ be a DFS ordering
of $T$ starting with $s$. Then $\sigma$ is a DFS ordering of $G$ with $\cl$-tree $T$. 
\end{lemma}

In this paper, we consider the leaves of search trees. For both $\cf$-trees and $\cl$-trees, we distinguish two different types of leaves.

\begin{definition}
Let $\sigma$ be a GS ordering of a connected graph $G$. A vertex $v \in V(G)$ is an \emph{\cf-leaf (\cl-leaf)} of $\sigma$ if $v$ is a leaf in the \cf-tree (\cl-tree) of $\sigma$. If $v$ is the first vertex of $\sigma$, then it is the \emph{\cf-root leaf (\cl-root leaf)} of $\sigma$, otherwise it is an \emph{\cf-branch leaf (\cl-branch leaf) of $\sigma$}.
\end{definition}

As the graph with exactly one vertex has no leaf in its spanning tree, we will consider only graphs with at least two vertices.

\section{Root Leaves}

We start this section with the simple observation that \cf-root leaves of GS orderings of a graph $G$ are quite boring as they are exactly the leaves of $G$.

\begin{observation}
Let $G$ be a connected graph with $n(G) \geq 2$. The following conditions are equivalent for a vertex $v \in V(G)$.
\hspace{1cm}
\begin{enumerate}[(i),leftmargin=1cm]
  \item Vertex $v$ is the \cf-root leaf of some GS ordering of $G$.
  \item Vertex $v$ is the \cf-root leaf of every GS ordering of $G$ starting in $v$.
  \item Vertex $v$ is a leaf of $G$.
\end{enumerate}
\end{observation}

Next we consider the $\cl$-root leaves of GS, DFS, and MCS. They are exactly those vertices of the graph that are not cut vertices. The same even holds for \cf-branch leaves and \cl-branch leaves of GS. 

\begin{theorem}\label{thm:l-root-leaf-all}
Let $G$ be a connected graph with $n(G) \geq 2$. The following conditions are equivalent for a vertex $v \in V(G)$.

\begin{enumerate}[(i)]
  \item Vertex $v$ is the $\cl$-root leaf of some DFS ordering of $G$ starting in $v$.\label{item:root-leaf-dfs-some}
  \item Vertex $v$ is the $\cl$-root leaf of every DFS ordering of $G$ starting in $v$.\label{item:root-leaf-dfs-all}
  \item Vertex $v$ is the $\cl$-root leaf of some MCS ordering of $G$.\label{item:root-leaf-mcs}
  \item Vertex $v$ is the $\cl$-root leaf of some GS ordering of $G$.\label{item:root-leaf-generic-lroot}
  \item Vertex $v$ is an $\cl$-branch leaf of some GS ordering of $G$.\label{item:root-leaf-generic-lbranch}
  \item Vertex $v$ is an $\cf$-branch leaf of some GS ordering of $G$.\label{item:root-leaf-generic-fbranch}
  \item Vertex $v$ is the end-vertex of some GS ordering of $G$.\label{item:root-leaf-generic-end}
  \item Vertex $v$ is not a cut vertex of $G$.\label{item:root-leaf-generic-cut}
\end{enumerate}
\end{theorem}

\begin{proof}
First we show the equivalence of all statements concerning GS and statement \ref{item:root-leaf-generic-cut}. Charbit et al.~\cite{charbit2014influence} showed that a vertex $v$ is an end-vertex of a GS ordering if and only if it is not a cut vertex of $G$. Thus, \ref{item:root-leaf-generic-end} $\Leftrightarrow$ \ref{item:root-leaf-generic-cut}. 
Clearly, the end-vertex of a GS ordering is also a branch leaf of its $\cf$-tree and its $\cl$-tree. Hence, \ref{item:root-leaf-generic-end} $\Rightarrow$ \ref{item:root-leaf-generic-lbranch}, \ref{item:root-leaf-generic-fbranch}.

Now assume that vertex $v$ is not a cut vertex of $G$. Then let $\sigma$ be a GS ordering of $G - v$ that starts in a neighbor $w$ of $v$. Obviously, $\sigma' = v \mdoubleplus \sigma$ is a GS ordering of $G$. As $G - v$ is connected, every vertex of $G-v$ accept from $w$ has a neighbor to its left in $\sigma$. Thus, the only neighbor of $v$ in the $\cl$-tree of $\sigma'$ is $w$ and $v$ is a GS \cl-root leaf of $G$. Hence, \ref{item:root-leaf-generic-cut} $\Rightarrow$ \ref{item:root-leaf-generic-lroot}.

Assume that $v$ is a cut vertex of $G$ and let $A$ and $B$ be two components of $G - v$.  If $v$ is the start vertex of the GS ordering $\sigma$ of $G$, then the leftmost vertices of $A$ and $B$ in $\sigma$ are neighbors of $v$ in the $\cl$-tree of $\sigma$. Therefore, $v$ is not the $\cl$-root leaf of $\sigma$ and \ref{item:root-leaf-generic-lroot} $\Rightarrow$ \ref{item:root-leaf-generic-cut}.

If the GS ordering $\sigma$ does not start with the cut vertex $v$, then w.l.o.g. we may assume that  $\sigma$ starts with a vertex of $A$. Then the parent of $v$ in the \cf-tree and in the \cl-tree of $\sigma$ is an element of $A$. Furthermore, the first vertex of $B$ in $\sigma$ is a child of $v$ in the $\cf$-tree and in the $\cl$-tree of $\sigma$. Therefore, $v$ is neither an $\cf$-branch leaf nor an $\cl$-branch leaf of a GS ordering of $G$. Hence, \ref{item:root-leaf-generic-lbranch} $\Rightarrow$ \ref{item:root-leaf-generic-cut} and \ref{item:root-leaf-generic-fbranch} $\Rightarrow$ \ref{item:root-leaf-generic-cut}.

As DFS orderings and MCS orderings are also GS orderings, the statements \ref{item:root-leaf-dfs-some} and \ref{item:root-leaf-mcs} directly imply \ref{item:root-leaf-generic-lroot} and, thus, they also imply \ref{item:root-leaf-generic-cut}.

To show that \ref{item:root-leaf-generic-cut} implies \ref{item:root-leaf-mcs}, assume that $v$ is not a cut vertex. Then let $\rho$ be a vertex ordering of $G$ that starts with $v$ and has all neighbors of $v$ to the right of all non-neighbors of $v$. Let $\sigma$ be the MCS\+$(\rho)$ ordering of $G$. Assume for contradiction that $v$ has more than one child in the \cl-tree $T$ of $\sigma$. Let $x$ be the leftmost child of $v$ in $\sigma$ and let $y$ be another child of $v$ in $T$. Since $T$ is an \cl-tree and $v$ is the first vertex of $\sigma$, both $x$ and $y$ have only one neighbor to its left in $\sigma$, namely $v$. As $v$ is not a cut vertex, there is an $x$-$y$-path $P$ in $G - v$. Let $z$ be the vertex nearest to $x$ on $P$ that is to the right of $y$ in $\sigma$. Vertex $z$ exists since $v$ is the only neighbor of $y$ to the left of $y$ in $\sigma$ and, thus, the neighbor of $y$ in $P$ is not to the left of $y$ in $\sigma$. Vertex $z$ can only have one neighbor to the left of $y$ since, otherwise, $\sigma$ would not be an MCS ordering. By the choice of $z$, this neighbor lies on $P$. Thus, vertex $z$ is not a neighbor of $v$ and $z$ is to the left of $y$ in $\rho$. This is a contradiction as $y$ and $z$ have both exactly one neighbor to the left of $y$ and, thus, MCS\+$(\rho)$ would have numbered $z$ instead of $y$. Hence, $v$ is an \cl-root leaf of the MCS ordering $\sigma$.

Statement \ref{item:root-leaf-dfs-all} trivially implies \ref{item:root-leaf-dfs-some}. It remains to show that \ref{item:root-leaf-generic-cut} implies \ref{item:root-leaf-dfs-all}. We show the contraposition. Let $\sigma$ be a DFS ordering starting with $v$ and let $T$ be the \cl-tree of $\sigma$. Assume $v$ has two children in $T$. Any path in $G$ between these two children runs through $v$, due to \cref{lemma:dfstrees}. Thus, $v$ is a cut vertex of $G$.
\end{proof}

Note that DFS differs from GS and MCS in this result. While for the latter three searches it is possible that a vertex is not the $\cl$-root leaf of a search ordering starting with that vertex, this is not possible for DFS.

Since DFS, LDFS, MCS, and MNS orderings are also GS orderings, \cref{thm:l-root-leaf-all} directly implies that we can characterize the $\cl$-leaves of these orderings.

\begin{theorem}\label{corol:l-leaves}
For any search $\A \in \{\text{GS, DFS, LDFS, MCS, MNS}\}$ and any vertex $v$ of a connected graph $G$ with $n(G) \geq 2$, the following statements are equivalent. 

\begin{enumerate}[(i)]
\item Vertex $v$ is the $\cl$-root leaf of some $\A$-ordering of $G$.
\item Vertex $v$ is an $\cl$-leaf of some $\A$-ordering of $G$.
\item Vertex $v$ is not a cut vertex of $G$.
\end{enumerate}
\end{theorem}

As we can check in linear time whether a vertex is a cut vertex, we can also recognize $\cl$-leaves of GS, DFS, LDFS, MCS, and MNS within this time bound. However, we will see in \cref{corol:l-branch-leaves-dfs} that at least for DFS the recognition of \cl-branch leaves is \NP-complete.

The characterization of \cl-root leaves given in \cref{thm:l-root-leaf-all} does not work for BFS as the following theorem shows.

\begin{theorem}\label{thm:root-leaf-bfs}
Let $G$ be a connected graph with $n(G) \geq 2$. A vertex $v \in V(G)$ is the $\cl$-root leaf of some BFS ordering of $G$ if and only if $G[N_G(v)]$ is connected.
\end{theorem}

\begin{proof}
First assume that $G[N_G(v)]$ is connected. Let $\sigma$ be a GS ordering of $G[N_G(v)]$. The ordering $v \mdoubleplus \sigma$ is a prefix of a BFS ordering of $G$ and all vertices of $N_G(v)$ except from the first in $\sigma$ have a neighbor to its left in $\sigma$. Thus, $v$ has only one child in the \cl-tree of every BFS ordering starting with $v \mdoubleplus \sigma$.

Now assume $G[N_G(v)]$ is not connected. Let $A$ and $B$ be two components of $G[N_G(v)]$ and let $\sigma$ be an arbitrary BFS ordering of $G$ starting with $v$. The leftmost vertices of $A$ and $B$ in $\sigma$ are neighbors of $v$ and, thus, they are children of $v$ in the \cl-tree of $\sigma$. Therefore, $v$ is not an \cl-root leaf of $\sigma$.
\end{proof}

\section{\NP-Hardness of Branch Leaf Recognition}

\subsection{Branch Leaves of DFS}
DFS \cl-trees can be recognized in linear time~\cite{hagerup1985biconnected,korach1989dfs}. As we have seen in \cref{corol:l-leaves}, this also holds for DFS $\cl$-leaves. In contrast, recognizing DFS \cl-branch leaves of a graph is as hard as the recognition of DFS end-vertices since the two concepts are equivalent.

\begin{theorem}\label{thm:end-leaf-dfs}
A vertex $v \in V(G)$ of a graph $G$ is an $\cl$-branch leaf of some DFS ordering of $G$ if and only if $v$ is the end-vertex of some DFS ordering of $G$.
\end{theorem}

\begin{proof}
If $v$ is the end-vertex of some DFS ordering $\sigma$ of $G$, then it is also a branch leaf of the $\cl$-tree of $\sigma$. 

For the other direction, assume that $v$ is a branch leaf of the $\cl$-tree $T$ of the DFS ordering $\sigma$ of $G$ starting with the vertex $r$. Consider a DFS ordering $\tau$ of $T$ starting with $r$ with the following constraint. Whenever a vertex has more than one unnumbered child, then the children that are not ancestors of $v$ are to the left of $v$ in $\sigma$. This means that the DFS always numbers non-ancestors of $v$ before ancestors of $v$ if this is possible. Due to \cref{lemma:tree}, $\sigma$ is also a DFS ordering of $G$.

We claim that $v$ is the end-vertex of $\sigma$. Assume for contradiction that $w \neq v$ is the end-vertex of $\sigma$. Then, $w$ is a branch leaf of $T$. Therefore, $w$ is not an ancestor of $v$. Let $x$ be the common ancestor of $v$ and $w$ that has the largest distance to the root vertex $r$ in $T$. Let $v'$ be the child of $x$ that is an ancestor of $v$ and let $w'$ be the child of $x$ that is an ancestor of $w$ ($v'$ could be equal to $v$ and $w'$ could be equal to $w$). By construction of $\sigma$, the vertex $w'$ is to the left of $v'$ in $\sigma$. Since the descendants of $w'$ appear consecutively directly after $w'$ in $\sigma$, $w$ is to the left of $v$ in $\sigma$; a contradiction.
\end{proof}

Charbit et al.~{\cite{charbit2014influence}} gave sufficient conditions on a graph class $\G$ such that the end-vertex problem of DFS is $\NP$-complete on $\G$. Due to \cref{thm:end-leaf-dfs}, we can replace the term end-vertex in their result by the term \cl-branch leaf.

\begin{corollary}\label{corol:l-branch-leaves-dfs}
Let $\G$ be a graph class that is closed under the insertion of universal vertices. If the Hamiltonian path problem is \NP-complete on $\G$, then the problem of deciding whether a vertex of a graph $G \in \G$ is an \cl-branch leaf of some DFS ordering of $G$ is \NP-complete. In particular, the problem is \NP-complete on split graphs.
\end{corollary}

A similar result can be given for \cf-branch leaves of DFS. By adapting the proof given in~\cite{scheffler2022recognition} that \cf-trees of DFS are hard to recognize, we can show that the same holds for \cf-branch leaves of DFS.

\begin{theorem}\label{thm:f-branch-leaves-dfs}
Let $\G$ be a graph class that is closed under the insertion of universal vertices and leaves. If the Hamiltonian path problem is \NP-complete on $\G$, then the problem of deciding whether a vertex of a graph $G \in \G$ is an \cf-branch leaf of some DFS ordering of $G$ is \NP-complete. In particular, the problem is \NP-complete on chordal graphs.
\end{theorem}

\begin{proof}
Let $G$ be a graph in $\G$ with the vertex set $\{v_1, \ldots, v_n\}$. W.l.o.g. we may assume that $n \geq 2$. We construct the graph $G'$ as follows: First we add the vertex $v'_i$ for all $i \in \{1, \ldots, n\}$ and connect it to $v_i$. Then we add the vertex $y$ which is adjacent to all other vertices. Due to the conditions on the graph class $\G$, the graph $G'$ is in $\G$. We claim that there is a DFS ordering of $G'$ with $\cf$-branch leaf $y$ if and only if there is a Hamiltonian path in $G$. 

If there is an Hamiltonian path in $G$, then there is a DFS ordering of $G'$ that starts with this path. In the \cf-tree of such an ordering, vertex $y$ is a branch leaf. 

Now assume that $y$ is a branch leaf of the \cf-tree $T$ of some DFS ordering $\sigma$ of $G$. Assume for contradiction that $v \in V(G)$ is to the right of $y$ in $\sigma$. As $v'$ is not a child of $y$ in $T$, vertex $v'$ is to the left of $y$ in $\sigma$. Thus, $v'$ is to the left of both its neighbors in $\sigma$ and, thus, $v'$ must be the start vertex of $\sigma$ with $s \neq x_i$ for all $i \in \{1,\ldots,k+3\}$. However, this implies that $y$ is the second vertex of $\sigma$ and, since $n \geq 2$, vertex $y$ has at least one child in $T$; a contradiction. It follows that all vertices $v \in V(G)$ are to the left of $y$ in $\sigma$. Now let $T'$ be the \cl-tree of $\sigma$. Due to \cref{lemma:dfstrees}, all the vertices of $G$ are ancestors of $y$ in $T'$ and, thus, they all lie on a path $P$ from the root of $T'$ to $y$. If some vertex $v'$ lies on this path, then either $v'$ is the first vertex of $P$ or the successor of $v'$ in $\sigma$ is vertex $y$. Thus, only the first and the last vertex of the path $P$ could be one of these vertices. Therefore, $P$ contains a Hamiltonian path of $G$. 

The Hamiltonian path problem is \NP-complete on chordal graphs~\cite{bertossi1986hamiltonian}. Furthermore, chordal graphs are closed under the addition of leaves and universal vertices as neither universal vertices nor leaves can be an element of an induced cycle of length $\geq 4$.
\end{proof}

If we compare \cref{corol:l-branch-leaves-dfs,thm:f-branch-leaves-dfs}, then we see that for $\cl$-branch leaves it is sufficient that the graph class $\G$ is closed under the addition of universal vertices while for $\cf$-branch leaves we have the additional condition that $\G$ is closed under the addition of leaves. We cannot omit this constraint (unless $\P = \NP$) as the \cf-branch leaf recognition problem of DFS can be solved in polynomial time on split graphs (see \cref{corol:algo-f-branch-dfs-split}).

\subsection{Branch Leaves of BFS} The end-vertex problem of BFS is \NP-complete, even if the graph is bipartite and the start vertex of the BFS ordering is fixed~\cite{charbit2014influence}. This fact can be used to show that recognizing BFS \cf-branch leaves is also \NP-complete. 

\begin{theorem}\label{thm:f-branch-leaves-bfs}
It is \NP-complete to decide whether a vertex of a bipartite graph $G$ is an \cf-branch leaf of some BFS ordering of $G$.
\end{theorem}

\begin{proof}
We reduce the beginning-end-vertex problem of BFS on bipartite graphs to the respective \cf-branch leaf recognition problem. Given a connected bipartite graph $G$ and a vertex $r \in V(G)$, the problem asks whether a vertex $v \in V(G) \setminus \{r\}$ can be the end-vertex of some BFS ordering of $G$ starting with $r$. The problem is \NP-complete~\cite{charbit2014influence}.

Let $G$ and $r,v \in V(G)$ be an input of this problem. Let $k$ be the eccentricity of $r$. W.l.o.g. we may assume that $v$ is in $N^k_G(r)$ since, otherwise, $v$ is obviously not the end-vertex of any BFS ordering of $G$ starting with $r$. We construct the graph $G'$ from $G$ as follows (see \cref{fig:bfs-f-np}). We add a path $P = (r, x_1, \ldots, x_{k+3})$ of length $k+3$ to $G$ as well as the edge $vx_{k+3}$. Furthermore, for every vertex $w \in N^k_G(r) \setminus \{v\}$ we add a vertex $w'$ and the edges $vw'$ and $ww'$ to $G$. We collect these vertices $w'$ in the set $L$. Note that $G'$ is also bipartite.

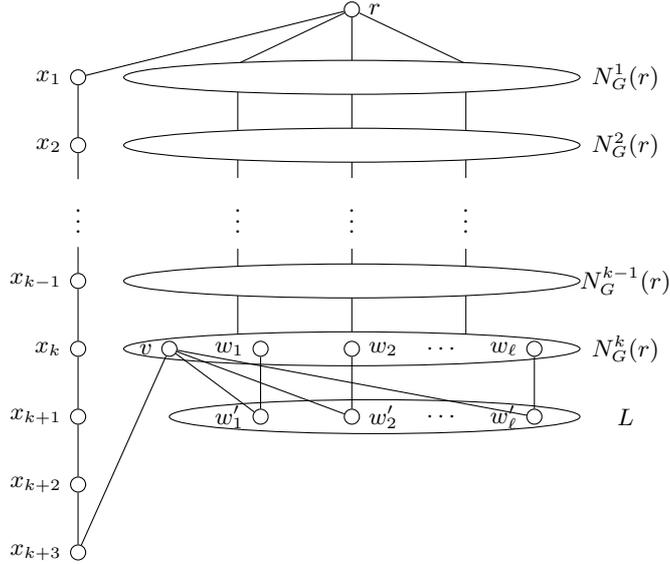
\begin{figure}[t]
\centering
\begin{tikzpicture}[vertex/.style={inner sep=2pt,draw,circle}, path/.style={decoration={snake, amplitude=0.3mm}, decorate}, edge/.style={-}, noedge/.style={dotted},xscale=1.2, yscale=0.9]
\footnotesize
\node[vertex,label={0:$r$}]  (r) at (0,0) {};
\draw (r) -- ++(-1.25, -0.8) --++ (0,-4);
\draw (r) -- ++(0, -0.75) --++ (0,-4);
\draw (r) -- ++(1.25, -0.8) --++ (0,-4);

\draw[fill=white] (0,-1) ellipse (2.5cm and 0.25cm);
\node[]  (d) at (3,-1) {$N_G^1(r)$};
\draw[fill=white] (0,-2) ellipse (2.5cm and 0.25cm);
\node[]  (d) at (3,-2) {$N_G^2(r)$};
\draw[fill=white] (0,-4) ellipse (2.5cm and 0.25cm);
\node[]  (d) at (3,-4) {$N_G^{k-1}(r)$};
\draw[fill=white] (0,-5) ellipse (2.5cm and 0.25cm);
\node[]  (d) at (3,-5) {$N_G^k(r)$};
\draw[fill=white] (0.25,-6) ellipse (2.25cm and 0.25cm);
\node[]  (d) at (3,-6) {$L$};

\node[fill=white,inner sep=2mm]  (d) at (-1.25,-3) {$\vdots$};
\node[fill=white,inner sep=2mm]  (d) at (1.25,-3) {$\vdots$};

\node[vertex,label={180:$x_1$}]  (x1) at (-3,-1) {};
\node[vertex,label={180:$x_2$}]  (x2) at (-3,-2) {};
\node[]  (d) at (-3,-3) {$\vdots$};
\node[fill=white,inner sep=2mm]  (d) at (0,-3) {$\vdots$};
\node[vertex,label={180:$x_{k-1}$}]  (xk-1) at (-3,-4) {};
\node[vertex,label={180:$x_k$}]  (xk) at (-3,-5) {};
\node[vertex,label={180:$x_{k+1}$}]  (xk1) at (-3,-6) {};
\node[vertex,label={180:$x_{k+2}$}]  (xk2) at (-3,-7) {};
\node[vertex,label={180:$x_{k+3}$}]  (xk3) at (-3,-8) {};

\node[vertex,label={180:$v$}]  (v) at (-2,-5) {};
\node[vertex,label={180:$w_1$}]  (w1) at (-1,-5) {};
\node[vertex,label={0:$w_2$}]  (w2) at (0,-5) {};
\node[]  (d) at (1,-5) {$\cdots$};
\node[vertex,label={180:$w_\ell$}]  (wl) at (2,-5) {};
\node[vertex,label={180:$w'_1$}]  (w'1) at (-1,-6) {};
\node[vertex,label={0:$w'_2$}]  (w'2) at (0,-6) {};
\node[]  (d) at (1,-6) {$\cdots$};
\node[vertex,label={180:$w'_\ell$}]  (w'l) at (2,-6) {};

\draw (r) -- (x1) -- (x2) -- ++(0,-0.5);
\draw (xk-1) -- (xk) -- (xk1) -- (xk2) -- (xk3);
\draw (xk-1) -- ++(0,0.5);
\draw (w'1) -- (v) -- (w'2);
\draw (v) -- (w'l);
\draw (w1) -- (w'1);
\draw (w2) -- (w'2);
\draw (wl) -- (w'l);

\draw (v) -- (xk3);

\end{tikzpicture}
\caption{A graphical representation of the graph $G'$ constructed in the proof of \cref{thm:f-branch-leaves-bfs}.}\label{fig:bfs-f-np}
\end{figure}

First assume that $v$ is the end-vertex of some BFS ordering $\sigma$ of $G$ starting with $r$. Let $\rho = (x_1,x_2, \ldots, x_{k+3}) \mdoubleplus \sigma$ and let $\tau$ be the BFS\+($\rho$) ordering of $G'$. By the choice of $\rho$, the ordering $\tau$ starts with $x_1$. Thus, every vertex $x_i$ is in layer $i-1$ and every vertex $v \in N^j_G(r)$ is in layer $j+1$. The choice of $\rho$ also implies that every vertex $x_i$ is the first vertex of its layer in $\tau$. In particular, vertex $x_{k+2}$ is to the left of $v$ in $\tau$ and, thus, $x_{k+2}$ is the parent of $x_{k+3}$ in the \cf-tree of $\tau$. Furthermore, $x_{k+3}$ is to the right of $v$ in $\tau$ since $dist_{G'}(x_1,x_{k+3}) = k +2$ and $dist(x_1, v) = k + 1$. For the same reason, the vertices $w' \in L$ are also to the right of any vertex of $V(G)$ in $\tau$. Therefore, the restriction of $\tau$ to the vertices of the graph $G$ is exactly $\sigma$ and, thus, every $w \in N^k_G(r) \setminus \{v\}$ is to the left of $v$ in $\tau$. Hence, $v$ is an \cf-branch leaf of $\tau$. 

Now assume that $v$ is an \cf-branch leaf of some BFS ordering $\sigma$ of $G'$. First we show that $\sigma$ starts with a vertex $x_i$. Assume for contradiction that this is not the case and let $s$ be the start vertex of $\sigma$. If every shortest path from $s$ to $x_{k+3}$ runs through $v$, then $v$ is the parent of $x_{k+3}$ in the \cf-tree of $\sigma$; a contradiction. Thus, we may assume that $dist_{G'}(s,x_{k+2}) < dist_{G'}(s,x_{k+3})$. This implies that $dist_{G'}(s,x_{k+2}) = dist_{G}(s,r) + k + 2$. However, $dist_{G'}(s,v) \leq dist_{G}(s,r) + dist_G(r,v) = dist_G(s,r) + k$. Hence, $v$ is to the left of $x_{k+2}$ in $\sigma$ and, therefore, $v$ is the parent of $x_{k+3}$ in the \cf-tree of $\sigma$; a contradiction. 

Thus, we may assume that $\sigma$ starts with some vertex $x_i$. For any vertex $w' \in L$, it holds $dist_{G'}(x_i,w') > dist_{G'}(x_i,v)$. Thus, all the elements of $L$ are to the right of $v$ in $\sigma$. Therefore, all the vertices in $N_G^k(r)$ are to the left of $v$ in $\sigma$ since $v$ is an \cf-branch leaf of $\sigma$. If there is any vertex $y \in V(G) \setminus N_G^k(r)$ to the right of $v$ in $\sigma$, then $dist_{G'}(x_i,y) \geq dist_{G'}(x_i,v)$. Thus, the shortest path from $x_i$ to $v$ runs through $x_{k+3}$ which implies that $dist_{G'}(x_i,x_{k+3}) < dist(x_i, z)$ for all $z \in N_G^{k-1}(r)$. However, as $x_{k+3}$ is adjacent to $v$ in $G$ but not adjacent to any other element of $N_G^k(r)$, vertex $v$ is the leftmost vertex of $N_G^k(r)$ in $\sigma$; a contradiction to the observation above.

Hence, $v$ is the rightmost vertex of $V(G)$ in $\sigma$. Since $\sigma$ starts with some $x_i$, vertex $r$ is the leftmost vertex of $V(G)$ in $\sigma$. Let $\sigma^*$ be the restriction of $\sigma$ to the vertices of $G$. The ordering $\sigma^*$ starts with $r$ and ends with $v$. None of the vertices in $V(G') \setminus V(G)$ had an influence on the order of the vertices of $G$ in $\sigma$. Thus, $\sigma^*$ is a BFS ordering of $G$ starting with $r$ and ending with $v$.
\end{proof}

In contrast to this result, there is a simple characterization of BFS \cl-branch leaves of bipartite graphs.

\begin{theorem}\label{thm:l-branch-leaves-bfs-bipartite}
Let $G$ be a connected bipartite graph with $n(G) \geq 2$. A vertex $v \in V(G)$ is an \cl-branch leaf of some BFS ordering of $G$ if and only if there is an $r \in V(G) \setminus \{v\}$ such that $dist_G(r,w) = dist_{G-v}(r,w)$ for all $w \in V(G) \setminus \{v\}$.
\end{theorem}

\begin{proof}
Assume that there is a vertex $r \in V(G) \setminus \{v\}$ such that $dist_G(r,w) = dist_{G-v}(r,w)$ for all $w \in V(G) \setminus \{v\}$. Let $(r = w_0, \ldots, w_k = v)$ be a shortest path from $r$ to $v$, i.e., $v$ has distance $k$ to $r$. It is easy to see that there is a BFS ordering $\sigma$ of $G$ in which every vertex $w_i$ is the first vertex of the $i$-th layer. Let $T$ be the \cl-tree of $\sigma$ and let $x$ be a vertex in the $(k+1)$-th layer. Due to the condition on $r$, there is a shortest path from $r$ to $x$ in $G$ that does not use vertex $v$. Therefore, $x$ has a neighbor $y$ in the $k$-th layer that is not $v$. Since $v \prec_\sigma y \prec_\sigma x$, vertex $v$ is not the parent of $x$ in $T$. Since $G$ is bipartite, the layers of $\sigma$ are independent sets and, thus, $v$ is neither the parent of any vertex in the $k$-th layer. Hence, $v$ is a leaf of $T$.

Now assume that $v$ is a branch leaf of the \cl-tree $T$ of the BFS ordering $\sigma$. Let $r$ be the start vertex of $\sigma$. Let $w$ be a vertex different from $v$ and $r$. Consider the $r$-$w$-path $P$ in $T$. Since $G$ is bipartite, the edges of $G$ and, thus, the edges of $T$ only connect vertices of consecutive layers. Furthermore, every vertex has a neighbor in its preceding layer. Thus, $P$ has $dist_G(r,w)$ edges. Since $v$ is a leaf of $T$, $P$ does not contain $v$. Therefore, $P$ is also contained in $G - v$ and $dist_{G-v}(r,w) = dist_{G}(r,w)$.
\end{proof}

To check whether the condition of \cref{thm:l-branch-leaves-bfs-bipartite} is fulfilled, we simply make two all-pair-shortest paths computations and compare the results. This can be done in $\O(n(G) \cdot m(G))$ by using $\O(n(G))$ many BFS computations.

\begin{corollary}\label{corol:l-branch-leaves-bfs-bipartite-algo}
Given a connected bipartite graph $G$ and a vertex $v \in V(G)$, we can decide in time $\O(n(G) \cdot m(G))$ whether $v$ is the \cl-branch leaf of some BFS ordering of $G$.
\end{corollary}

The results of \cref{thm:f-branch-leaves-bfs,corol:l-branch-leaves-bfs-bipartite-algo} are quite surprising since the \cl-tree recognition problem of BFS is \NP-hard on bipartite graphs~\cite{scheffler2022recognition} while the \cf-tree recognition problem of BFS can be solved in linear time~\cite{hagerup1985recognition,manber1990recognizing}.

\subsection{Branch Leaves of MNS-like Searches}
For several subsearches of MNS, the recognition problem of \cf-branch leaves is \NP-complete on weakly chordal graphs.

\begin{theorem}\label{thm:f-branch-mns}
Let $\A$ be one of the following searches: LBFS, LDFS, MCS, MNS. It is \NP-complete to decide whether a vertex of a weakly chordal graph $G$ is an \cf-branch leaf of some $\A$-ordering.
\end{theorem}

\begin{proof}
The proof of the theorem is inspired by the \NP-completeness proof of the \cf-tree recognition problem of MNS given by Beisegel~et~al.~\cite{beisegel2021recognition}. We construct a polynomial-time reduction from 3-SAT. Let $\I$ be an instance of 3-SAT. W.l.o.g.~we may assume that $\I$ contains at least two clauses. We construct the corresponding graph $G(\I)$ as follows. Let $ X=\{x_1, \dots, x_k,\overline{x}_1,\ldots,\overline{x}_k\} $ be the set of vertices representing the literals of $ \I $. The graph $G(\I)[X]$ forms the complement of the matching in which $x_i$ is matched to $ \overline{x}_i $ for every $ i \in \{1,\ldots, k\} $. Let $ C = \{c_1, \ldots ,c_\ell\} $ be the set of vertices representing the clauses of $ \I $. The set $ C $ forms an independent set in $G(\I)$ and every clause vertex $ c_i $ is adjacent to each vertex of $ X $ whose corresponding literal is contained in the clause associated with $ c_i $. Additionally, we add a universal vertex $t$.

Assume $G(\I)$ has a fulfilling assignment $\B$. Then we create the following \A-ordering $\sigma$. We first number all literal vertices of literals that are set to true in $\B$ and then we number $t$. Since these vertices form a clique, this ordering is a prefix of an $\A$-ordering. We number the remaining vertices following an arbitrary $\A$-ordering. As $\B$ is fulfilling, all clause vertices and all literal vertices have a neighbor that is to the left of $t$ in $\sigma$. Thus, $t$ is an \cf-branch leaf of $\sigma$.

Now assume that $t$ is an \cf-branch leaf of the $\A$-ordering $\sigma$ of $G(\I)$. Let $S$ be the set of literal vertices that are to the left of $t$ in $\sigma$. Since $t$ is universal and the edges $x_i\overline{x}_i$ are missing, the set $S$ contains at most one literal vertex for every variable. Thus, we can define an assignment $\B$ by giving all literals whose vertices are contained in $S$ the value true. If some variable value is not fixed, then we choose an arbitrary value for the variable. If a clause vertex has a parent in the \cf-tree $T$ of $\sigma$, then this parent is an element of $S$ since $t$ is a leaf in $T$. If the clause vertex $c_i$ does not have a parent in $T$, then $c_i$ is the first vertex of $\sigma$. Since there are at least two clause vertices, the second vertex of $\sigma$ is not $t$ but a literal vertex adjacent to $c_i$. Therefore, every clause vertex has a neighbor in $S$ and, thus, $\B$ is a fulfilling assignment.

To see that $G(\I)$ is weakly chordal, we first observe that every pair $(x_i, \overline{x}_i)$ forms a two-pair in $G(\I)$. Spinrad and Sritharan~\cite{spinrad1995algorithms} showed that the graph that results from the addition of an edge between a two-pair is weakly chordal if and only if the initial graph is weakly chordal. If we add all the edges $x_i\overline{x}_i$, then the resulting graph is a split graphs and, thus, $G(\I)$ is weakly chordal.
\end{proof}

\section{Branch Leaves and Chordal Graphs}
\subsection{Branch Leaves of MNS-like Searches}

MNS and all of its subsearches compute PEOs of chordal graphs~\cite{corneil2008unified,tarjan1984simple}. Thus, any \cf-tree or \cl-tree of an MNS ordering is also an \cf-tree or \cl-tree of some PEO. Beisegel et al.~\cite{beisegel2021recognition} showed that this also holds the other way around for a large family of graph searches including LBFS, LDFS, MCS, and MNS, i.e., the rooted \cf-trees and rooted \cl-trees of these searches on chordal graphs are exactly the rooted \cf-trees and rooted \cl-trees of PEOs, respectively. Therefore, we will only characterize \cf-branch leaves and \cl-branch leaves of PEOs.

We start by showing that the \cl-branch leaves of PEOs of a chordal graph are exactly the graph's simplicial vertices.

\begin{theorem}\label{thm:l-branch-mns-chordal}
Let $G$ be a connected chordal graph with $n(G) \geq 2$. A vertex $v \in V(G)$ is an $\cl$-branch leaf of some PEO of $G$ if and only if $v$ is simplicial. 
\end{theorem}

\begin{proof}
If $v$ is a simplicial vertex, then there is a PEO $\sigma$ that ends with $v$. Vertex $v$ is an $\cl$-branch leaf of $\sigma$.

For the other direction, let $\sigma$ be a PEO and let $v$ be a non-simplicial vertex of $G$. Hence, not all neighbors of $v$ are to the left of $v$ in $\sigma$. Let $w$ be the leftmost neighbor of $v$ in $\sigma$ that is to the right of $v$ in $\sigma$. Let $x$ be the parent of $w$ in the \cl-tree $T$ of $\sigma$. If $x$ is not equal $v$, then it holds $v \prec_\sigma x \prec_\sigma w$. As $\sigma$ is a PEO and $vw,xw \in E(G)$, the edge $vx$ is also in $E(G)$; a contradiction to the choice of $w$. Hence, $v$ is the parent of $w$ in $T$ and $v$ is not an \cl-branch leaf of $\sigma$.
\end{proof}

Since we can decide in linear time whether a vertex is simplicial~\cite{beisegel2019end-vertex}, we can recognize \cl-branch leaves of PEOs in linear time.

\begin{corollary}\label{corol:algo-l-branch-mns-chordal}
Given a connected chordal graph $G$ and a vertex $v \in V(G)$, we can decide in time $\O(n(G) + m(G))$ whether $v$ is the \cl-branch leaf of some PEO of $G$. Therefore, we can also decide in time $\O(n(G) + m(G))$ whether $v$ is the \cl-branch leaf of some LBFS, LDFS, MCS, or MNS ordering.
\end{corollary}

Obviously, simplicial vertices are also \cf-branch leaves of PEOs. However, there are further \cf-branch leaves.

\begin{theorem}\label{thm:f-branch-mns-chordal}
Let $G$ be a connected chordal graph with $n(G) \geq 2$. A vertex $v \in V(G)$ is an $\cf$-branch leaf of some PEO of $G$ if and only if the graph $G[N_G(v)]$ has a dominating clique.
\end{theorem}

\begin{proof}
First assume that $G[N_G(v)]$ has a dominating clique $C$. It is obvious that there is an LBFS ordering $\sigma$ of $G$ that starts with the vertices of $C$. The ordering $\sigma$ is a PEO. Since all neighbors of $v$ have a neighbor in $C$ or are elements of $C$, $v$ is an \cf-branch leaf of $\sigma$. 

Now let $v$ be an $\cf$-branch leaf of the PEO $\sigma$. Let $S$ be the set of neighbors of $v$ that are to the left of $v$ in $\sigma$. The set $S$ induces a clique of $G$. Thus, if $S = N_G(v)$, then $G[N_G(v)]$ is a clique and we are done. Hence, we may assume that there is a vertex $w \in N_G(v) \setminus S$. As $w$ is not a child of $v$ in the $\cf$-tree of $\sigma$, there is a vertex $x \in N_G(w)$ with $x \prec_\sigma v$. Since $\sigma$ is a PEO, vertex $x$ is a neighbor of $v$ and, thus, $x \in S$. Thus, any neighbor of $v$ that is not in $S$ has a neighbor in $S$ and, hence, $S$ induces a dominating clique of $G[N_G(v)]$. 
\end{proof}

To decide the complexity of the \cf-branch leaf recognition problem of PEOs, we examine the complexity of deciding the existence of a dominating clique in a chordal graph. Kratsch~et~al.~\cite{kratsch1994dominating} showed that such a clique exists if and only if the diameter of the graph is at most three.

\begin{theorem}[Kratsch et al.~{\cite{kratsch1994dominating}}]\label{thm:chordal-dominating-clique}
A chordal graph $G$ has a dominating clique if and only if the diameter of $G$ is at most three.
\end{theorem}

As the diameter of a graph can be determined by computing $n(G)$ many BFS orderings, we can decide the existence of a dominating clique in a chordal graph in polynomial time. Although it is unlikely that the diameter of a chordal graph can be computed in linear time,\footnote{Even on split graphs, the diameter cannot be computed in subquadratic time unless the Strong Exponential Time Hypothesis fails~\cite{borassi2016square}.}
we can improve our algorithm to decide the existence of a dominating clique in linear time. To this end, we can use the following result of Corneil et al.~\cite{corneil2001diameter}.

\begin{theorem}[Corneil et al.~{\cite{corneil2001diameter}}]\label{thm:chordal-diameter}
Let $G$ be a chordal graph and let $v \in V(G)$ be the end-vertex of some LBFS ordering of $G$. If $ecc(v) < diam(G)$, then $ecc(v)$ is even and $ecc(v) = diam(G) - 1$.
\end{theorem}

Combining \cref{thm:f-branch-mns-chordal,thm:chordal-dominating-clique,thm:chordal-diameter}, we can give a linear-time recognition algorithm for \cf-branch leaves of PEOs.

\begin{corollary}\label{corol:algo-f-branch-mns-chordal}
Given a chordal graph $G$ and a vertex $v \in V(G)$, we can decide in time $\O(n(G) + m(G))$ whether $v$ is an \cf-branch leaf of some PEO of $G$. Therefore, we can also decide in time $\O(n(G) + m(G))$, whether $v$ is the \cf-branch leaf of some LBFS, LDFS, MCS, or MNS ordering.
\end{corollary}

\begin{proof}
Due to \cref{thm:f-branch-mns-chordal,thm:chordal-dominating-clique}, it is sufficient to check whether $G' = G[N_G(v)]$ has diameter 3. We compute an LBFS ordering $\sigma$ of $G'$ in linear time. Let $v$ be the end-vertex of $\sigma$. We compute the eccentricity of $v$ in $G'$ in linear time by starting a BFS in $v$. If $ecc_{G'}(v) > 3$, then the diameter of $G'$ is larger than $3$. If $ecc_{G'}(v) = 3$, then, by \cref{thm:chordal-diameter}, $diam(G') = 3$. If $ecc_{G'}(v) < 3$, then $diam(G') \leq ecc_{G'}(v) + 1 \leq 3$, due to \cref{thm:chordal-diameter}.
\end{proof}

\subsection{Branch Leaves of BFS}

\begin{figure}[t]
\centering
\begin{tikzpicture}[vertex/.style={inner sep=2pt,draw,circle}, path/.style={decoration={snake, amplitude=0.3mm}, decorate}, edge/.style={-}, noedge/.style={dotted}, scale=1.25]
\footnotesize
\node[vertex,label={-90:$u$}]  (1) at (0,0) {};
\node[vertex,label={-90:$v$}]  (2) at (0.5,0) {};
\node[vertex,label={-90:$w$}]  (3) at (1,0) {};
\node[vertex,label={-90:$x$}]  (4) at (1.5,0) {};
\node[vertex,label={-90:$y$}]  (5) at (2,0) {};
\node[vertex,label={90:$z$}]  (6) at (1,0.5) {};

\draw[treeedge] (1) -- (2) -- (3) -- (4) -- (5);
\draw[treeedge] (3) -- (6);
\draw (5) -- (6) -- (4) -- (6) -- (2) -- (6) -- (1);
\end{tikzpicture}
\caption[BFS \cf-branch leaf of a chordal graph]{The given graph $G$ is chordal. There is no dominating clique in the graph $G[N_G(z)]$. However, the given spanning tree is the \cf-tree of the BFS ordering $(w,v,x,z,u,y)$ and, thus, $z$ is a \cf-branch leaf of BFS.}\label{fig:bfs-f-leaf-chordal}
\end{figure}
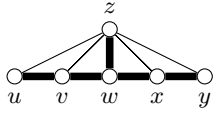

The condition given in \cref{thm:f-branch-mns-chordal} is also sufficient for a vertex to be a BFS \cf-branch leaf since every LBFS ordering is also a BFS ordering. However, it is not necessary as can be seen in \cref{fig:bfs-f-leaf-chordal}. To characterize BFS \cf-branch leaves of chordal graphs, we start with the following two lemmas.

\begin{lemma}\label{lemma:bfs-chordal1}
Let $G$ be a chordal graph and let $r$ be a vertex in $V(G)$. Let $x$ and $y$ be two vertices in $N_G^i(r)$. If there is a vertex $z \in N_G^{i+1}(r)$ which is adjacent to both $x$ and $y$, then $xy \in E(G)$. 
\end{lemma}

\begin{proof}
Let $\sigma$ be an LBFS ordering of $G$ starting with $r$. Vertex $z$ is to the right of both $x$ and $y$ in $\sigma$. Since $\sigma$ is a PEO, the vertices $x$ and $y$ are adjacent. 
\end{proof}

\begin{lemma}\label{lemma:bfs-chordal2}
Let $G$ be a chordal graph and let $r$ be a vertex in $V(G)$. Let $x$ and $y$ be two vertices in $N_G^i(r)$. If $xy \in E(G)$, then $N_G(x) \cap N_G^{i-1}(r) \subseteq N_G(y)$ or $N_G(y) \cap N_G^{i-1}(r) \subseteq N_G(x)$.
\end{lemma}

\begin{proof}
Let $\sigma$ be an LBFS ordering starting with $r$. W.l.o.g. we may assume that $x \prec_\sigma y$. All the vertices of $N_G^{i-1}(r)$ are to the left of $y$ in $\sigma$. Since $\sigma$ is a PEO, all the neighbors of $y$ in $N_G^{i-1}(r)$ are adjacent to $x$ and, thus, $N_G(y) \cap N_G^{i-1}(r) \subseteq N_G(x)$.
\end{proof}

The next lemma makes a statement about the distances of neighbors of a vertex in a chordal graph.

\begin{lemma}\label{lemma:bfs-chordal3}
Let $G$ be a connected chordal graph and let $v \in V(G)$. For any $x,y \in N_G(v)$, the distance between $x$ and $y$ in $G - v$ is equal to the distance between $x$ and $y$ in $G[N_G(v)]$. 
\end{lemma}

\begin{proof}
Let $P$ be a shortest $x$-$y$-path in $G - v$. If $P$ is a subgraph of $G[N_G(v)]$, then we are done. Otherwise, let $u$ be the first vertex of $P$ (starting from $x$) whose successor on $P$ is not in $N_G(v)$ and let $w$ be the first element of $P$ after $u$ that is an element of $N_G(v)$. Note that $u$ could be equal to $x$ and $w$ could be equal to $y$. The subpath of $P$ between $u$ and $w$ contains at least three vertices. Therefore, this subpath and the edges $uv$ and $vw$ form an induced cycle of length at least four in $G$; a contradiction.
\end{proof}

Using \cref{lemma:bfs-chordal1,lemma:bfs-chordal2,lemma:bfs-chordal3}, we characterize BFS \cf-branch leaves of chordal graphs.

\begin{theorem}\label{thm:f-branch-bfs-chordal}
Let $G$ be a connected chordal graph with $n(G) \geq 2$. A vertex $v \in V(G)$ is an $\cf$-branch leaf of some BFS ordering of $G$ if and only if the radius of $G[N_G(v)]$ is at most two.
\end{theorem}

\begin{proof}
First assume that $G[N_G(v)]$ has radius two and let $w$ be a central vertex of $G[N_G(v)]$. There is a BFS ordering $\sigma$ that starts with $w$ followed by all neighbors of $w$ that are not $v$. Vertex $v$ is an \cf-branch leaf of $\sigma$ since all neighbors of $v$ have some neighbor in $N_G[w] \setminus \{v\}$ or are equal to $w$. 

Now assume that $v \in N_G^i(r)$ is an \cf-branch leaf of some BFS ordering $\sigma$ of $G$ starting with $r$. Let $T$ be the \cf-tree of $\sigma$ rooted in $r$ and let $v'$ be the parent of $v$ in $T$. Since $T$ is an BFS \cf-tree rooted in $r$, it holds that $v' \in N^{i-1}_G(r)$. 

We claim that in $G - v$ vertex $v'$ has a distance of at most two to every element of $N_G(v)$. Let $w \in N_G(v) \setminus \{v'\}$. If $v'w \in E(G)$, then $v'$ and $w$ have distance one in $G - v$. Therefore, we may assume in the following that $v'w \notin E(G)$. Then \cref{lemma:bfs-chordal1} implies that $w \notin N_G^{i-1}(r)$. Furthermore, the parent of $w$ in $T$, say $w'$, is different from $v$ and $v'$. If $v'w' \in E(G)$, then $v'$ and $w$ have distance two in $G - v$ via the path $(v', w', w)$.  Thus, we may also assume that $v'w' \notin E(G)$.

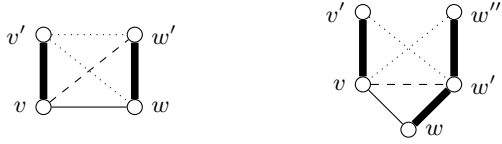
\begin{figure}[t]
\centering
\begin{tikzpicture}[vertex/.style={inner sep=2pt,draw,circle}, path/.style={decoration={snake, amplitude=0.3mm}, decorate}, edge/.style={-}, scale=0.6, noedge/.style={dotted}, every edge quotes/.style = {auto, font=\scriptsize, sloped}]
\footnotesize
\node[vertex,label={180:$v$}]  (v) at (-1,0) {};
\node[vertex,label={0:$w$}]  (w) at (1,0) {};
\node[vertex,label={180:$v'$}]  (v') at (-1,1.6) {};
\node[vertex,label={0:$w'$}]  (w') at (1,1.6) {};

\draw[treeedge] (v) -- (v');
\draw[treeedge] (w) -- (w');
\draw (v) -- (w);
\draw[noedge] (v') -- (w);
\draw[dashed] (v) -- (w');
\draw[dotted] (v') -- (w');

\begin{scope}[xshift=7cm, yshift=0.5cm]
\node[vertex,label={180:$v$}]  (v) at (-1,0) {};
\node[vertex,label={0:$w$}]  (w) at (0,-1) {};
\node[vertex,label={180:$v'$}]  (v') at (-1,1.6) {};
\node[vertex,label={0:$w'$}]  (w') at (1,0) {};
\node[vertex,label={0:$w''$}]  (w'') at (1,1.6) {};

\draw[treeedge] (v) -- (v');
\draw[treeedge] (w) -- (w') -- (w'');
\draw (v) -- (w);
\draw[noedge] (w'') -- (v);
\draw[dashed] (v) -- (w');
\draw[dotted] (v') edge (w');
\end{scope}
\end{tikzpicture}
\caption[BFS \cf-branch leaf of a chordal graph]{Two cases of the proof of \cref{thm:f-branch-bfs-chordal}. The vertical arrangement of the vertices represent their layers. Thick edges are edges of the \cf-tree. Dotted edges are not present. Dashed edges are implied by either \cref{lemma:bfs-chordal1} or \cref{lemma:bfs-chordal2}.}\label{fig:bfs-f-leaf-chordal-proof}
\end{figure}

First assume that $w \in N_G^{i}(r)$ (see left part of \cref{fig:bfs-f-leaf-chordal-proof}). Then $w' \in N_G^{i-1}(r)$. \Cref{lemma:bfs-chordal2} implies that $w'$ is adjacent to $v$ because $vw$, $vv'$, and $ww' \in E(G)$ but $v'w \notin E(G)$. Now the non-existence of $v'w'$ contradicts \cref{lemma:bfs-chordal1}. 

Now assume that $w \in N_G^{i+1}(r)$ (see right part of \cref{fig:bfs-f-leaf-chordal-proof}). Then, due to \cref{lemma:bfs-chordal1}, $vw' \in E(G)$. Since $v'w' \notin E(G)$, the parent of $w'$, say $w''$, is different from $v'$. Since $w'$ is the parent of $w$, it holds that $w' \prec_\sigma v$. This implies that $w'' \prec_\sigma v'$. Therefore, $w''$ is not adjacent to $v$ since, otherwise, $v'$ would not be the parent of $v$. The non-existence of both $v'w'$ and $vw''$ contradicts \cref{lemma:bfs-chordal2}.

Summarizing, $v'$ has distance at most two in $G - v$ to any neighbor of $v$. By \cref{lemma:bfs-chordal3}, $v'$ has distance at most two in $G[N_G(v)]$ to any neighbor of $v$ in $G$ and, thus, $G[N_G(v)]$ has radius at most two.
\end{proof}

Chepoi and Dragan~\cite{chepoi1994linear} presented a linear-time algorithm that computes a central vertex of a chordal graph. As the eccentricity of such a vertex can be computed in linear time using BFS, we can compute the radius of a chordal graph and, in particular, of $G[N_G(v)]$ in linear time. Thus, \cref{thm:f-branch-bfs-chordal} implies a linear-time algorithm for the BFS \cf-branch leaf recognition on chordal graphs.

\begin{corollary}\label{corol:algo-f-branch-bfs-chordal}
Given a connected chordal graph $G$ and a vertex $v \in V(G)$, we can decide in time $\O(n(G) + m(G))$ whether $v$ is an \cf-branch leaf of some BFS ordering of $G$.
\end{corollary}

\subsection{Branch Leaves of DFS}

As we have seen in \cref{thm:f-branch-leaves-dfs}, the \cf-branch leaf recognition problem of DFS is \NP-complete on chordal graphs. However, there is a simple characterization of DFS \cf-branch leaves of split graphs.

\begin{theorem}\label{thm:f-branch-bfs-dfs-split}
Let $G$ be a connected split graph with $n(G) \geq 2$. A vertex $v \in V(G)$ is an $\cf$-branch leaf of some DFS ordering if and only if $v$ is not a cut vertex of $G$.
\end{theorem}

\begin{proof}
Due to \cref{thm:l-root-leaf-all}, no cut vertex of a graph $G$ is an $\cf$-branch leaf of any GS ordering. 
It remains to show that every vertex $v$ of a connected split graph $G$ for which $G -v$ is connected is an $\cf$-branch leaf of some DFS ordering of $G$. For the connected split graph with two vertices, this is straightforwardly true. Thus, we may assume in the following that $n(G) \geq 3$. Let $(C,I)$ be a the partition of $V(G)$ into a clique $C$ and the independent set $I$. First note that $C$ must contain at least one vertex that is not $v$ since, otherwise, $v$ would be a cut vertex of $G$. Obviously, there is a DFS ordering $\sigma$ of $G$ starting with the vertices of $C \setminus \{v\}$. As $v$ is not a cut vertex, every vertex of $I$ has at least one neighbor in $C \setminus \{v\}$. Thus, $v$ is an \cf-branch leaf of $\sigma$.
\end{proof}

As cut vertices can be identified in linear time, \cref{thm:f-branch-bfs-dfs-split} leads directly to a linear-time algorithm for the DFS \cf-branch leaf recognition on split graphs.

\begin{corollary}\label{corol:algo-f-branch-dfs-split}
Given a connected split graph $G$ and a vertex $v \in V(G)$, we can decide in time $\O(n(G) + m(G))$ whether $v$ is an \cf-branch leaf of some DFS ordering of $G$.
\end{corollary}

In contrast to this result, it is \NP-hard to decide whether a vertex of a split graph is a DFS \cl-branch leaf (see \cref{corol:l-branch-leaves-dfs}). Thus, the \cl-branch leaf recognition of DFS seems to be harder than the \cf-branch leaf recognition of DFS, a surprising contrast to the hardness of the DFS \cf-tree recognition problem~\cite{scheffler2022recognition} and the easiness of the DFS \cl-tree recognition problem~\cite{hagerup1985biconnected,korach1989dfs}. Recall that we have made a similar observation for the complexity of the branch leaf and tree recognition of BFS (see \cref{thm:f-branch-leaves-bfs,corol:l-branch-leaves-bfs-bipartite-algo}).

\newpage

\bibliographystyle{plainurl}
\bibliography{leaf-recognition}

\end{document}